\newif\ifcomments       %
\newif\ifpreliminary    %
\newif\ifwatermark      %
\newif\ifexabs          %
\newif\ifanon           %
\newif\iflipics
\newif\ifintroonly
\newif\ifitcsproc

\iflipics
  \documentclass[a4paper,UKenglish,cleveref, autoref, thm-restate]{lipics-v2021}
\else
  \documentclass[letterpaper,11pt,pdfa]{article}
  \usepackage[in]{fullpage}
  
  \usepackage{iftex}
  \ifPDFTeX
    \usepackage[utf8]{inputenc}
    \usepackage[noTeX]{mmap}
    \usepackage[T1]{fontenc}
  \fi
  \ifLuaTeX
    \usepackage{luatex85}
    \usepackage[noTeX]{mmap}
  \fi
\fi

\usepackage{amsmath}
\usepackage{amsfonts}
\usepackage{amsthm}
\usepackage{amssymb}
\usepackage{braket}
\usepackage{color}
\iflipics\else
  \usepackage[bookmarks]{hyperref}
  \usepackage[nameinlink]{cleveref}
\fi

\usepackage{tikz}
\usetikzlibrary{arrows,positioning}

\iflipics\else
\ifexabs
  \newtheorem{theorem}{Theorem}
\else
  \newtheorem{theorem}{Theorem}[section]
\fi
\newtheorem{definition}[theorem]{Definition}

\newtheorem{lemma}[theorem]{Lemma}
\newtheorem{corollary}[theorem]{Corollary}

\fi
\newtheorem{fact}[theorem]{Fact}

\iflipics
\else
  \usepackage[style=alphabetic,minalphanames=3,maxalphanames=4,maxnames=99,backref=true]{biblatex}
  \renewcommand*{\labelalphaothers}{\textsuperscript{+}}
  \renewcommand*{\multicitedelim}{\addcomma\space}
  \DeclareFieldFormat{eprint:iacr}{Cryptology ePrint Archive: \href{https://ia.cr/#1}{\texttt{#1}}}
  \DeclareFieldFormat{eprint:iacrarchive}{Cryptology ePrint Archive: \href{https://eprint.iacr.org/archive/#1}{\texttt{#1}}}
  \AtEveryBibitem{%
   \clearlist{address}
   \clearfield{date}
   \clearlist{location}
   \clearfield{month}
   \clearfield{series}
   \ifentrytype{book}{}{%
    \clearlist{publisher}
    \clearname{editor}
    \clearfield{isbn}
    \clearfield{issn}
   }
  }
  
  \usepackage{enumitem}
  \setlist[description]{noitemsep}
  \setlist[enumerate]{noitemsep}
  \setlist[itemize]{noitemsep}
\fi

\ifwatermark
  \usepackage{draftwatermark}
  \SetWatermarkText{WORKING DRAFT: NOT FOR DISTRIBUTION}
  \SetWatermarkColor[gray]{0.5}
  \SetWatermarkFontSize{1cm}
  \SetWatermarkAngle{90}
  \SetWatermarkHorCenter{20cm}
\fi

\usepackage{soul}
\usepackage{xcolor}
\usepackage{xparse}
\makeatletter
\ExplSyntaxOn
\cs_new:Npn \white_text:n #1
{
  \fp_set:Nn \l_tmpa_fp {#1 * .01}
  \llap{\textcolor{white}{\the\SOUL@syllable}\hspace{\fp_to_decimal:N \l_tmpa_fp em}}
  \llap{\textcolor{white}{\the\SOUL@syllable}\hspace{-\fp_to_decimal:N \l_tmpa_fp em}}
}
\NewDocumentCommand{\whiten}{ m }
{
  \int_step_function:nnnN {1}{1}{#1} \white_text:n
}
\ExplSyntaxOff

\NewDocumentCommand{ \varul }{ D<>{5} O{0.2ex} O{0.1ex} +m } {%
  \begingroup
  \setul{#2}{#3}%
  \def\SOUL@uleverysyllable{%
    \setbox0=\hbox{\the\SOUL@syllable}%
    \ifdim\dp0>\z@
    \SOUL@ulunderline{\phantom{\the\SOUL@syllable}}%
    \whiten{#1}%
    \llap{%
      \the\SOUL@syllable
      \SOUL@setkern\SOUL@charkern
    }%
    \else
    \SOUL@ulunderline{%
      \the\SOUL@syllable
      \SOUL@setkern\SOUL@charkern
    }%
    \fi}%
  \ul{#4}%
  \endgroup
}
\makeatother

\newcommand{\E}{\mathop{\mathbb{E}}}

\usepackage{mleftright}
\newcommand{\paren}[1]{\left(#1\right)}
\newcommand{\mparen}[1]{\mleft(#1\mright)}
\newcommand{\bracket}[1]{\left[#1\right]}
\newcommand{\mbracket}[1]{\mleft[#1\mright]}
\newcommand{\abs}[1]{\left|#1\right|}
\newcommand{\norm}[1]{\left\lVert#1\right\rVert}
\newcommand{\IP}[1]{\left\langle#1\right\rangle}

\DeclareMathOperator*{\TraceDist}{TD}
\DeclareMathOperator*{\Tr}{Tr}

\newcommand{\ketbra}[2]{\ket{#1}\!\bra{#2}}
\newcommand{\cPpoly}{\mathsf{P/poly}}
\newcommand{\cBPP}{\mathsf{BPP}}
\newcommand{\cBQP}{\mathsf{BQP}}
\newcommand{\cBQPqpoly}{\mathsf{BQP/qpoly}}

\newcommand{\cQMA}{\mathsf{QMA}}
\newcommand{\cQSZK}{\mathsf{QSZK}}
\newcommand{\cQSZKHV}{\mathsf{QSZK_{HV}}}
\newcommand{\cQCZK}{\mathsf{QCZK}}
\newcommand{\cQCZKHV}{\mathsf{QCZK_{HV}}}
\newcommand{\cCZK}{\mathsf{CZK}}
\newcommand{\cQIP}{\mathsf{QIP}}
\newcommand{\cIP}{\mathsf{IP}}
\newcommand{\cP}{\mathsf{P}}
\newcommand{\cNP}{\mathsf{NP}}
\newcommand{\cPSPACE}{\mathsf{PSPACE}}

\ifitcsproc
\renewcommand{\paragraph}[2]{\smallskip\noindent{\bf #2\ \ }}
\fi

\ifcomments
  \newcommand{\luowen}[1]{\textcolor{blue}{Luowen: #1}}
  \newcommand{\ran}[1]{\textcolor{purple}{Ran: #1}}
  \newcommand{\znote}[1]{\textcolor{red}{{\large [} \textbf{Z:} #1 {\large]}}}
\else
  \newcommand{\luowen}[1]{}
  \newcommand{\ran}[1]{}
  \newcommand{\znote}[1]{}
\fi

\newcommand{\ignore}[1]{}
\newcommand{\acks}{
  LQ thanks Prabhanjan Ananth and Henry Yuen for the motivating discussions prior to this work.
  LQ is also grateful to Scott Aaronson and William Kretschmer for their insightful discussions for \Cref{sec:discussions}.
  Some of these discussions occurred while LQ was visiting the Simons Institute for the Theory of Computing.
  We thank Fermi Ma for his helpful comments on a preliminary version of this work.
}

\title{On the computational hardness needed for quantum cryptography
}
\iflipics
  \author{Zvika Brakerski}{Weizmann Institute of Science, Rehovot, Israel \and \url{https://zvikab.bitbucket.io}}{zvika.brakerski@weizmann.ac.il}{https://orcid.org/0000-0002-4867-7999}{Supported by the Israel Science Foundation (Grant No.\ 3426/21), and by the European Union Horizon 2020 Research and Innovation Program via ERC Project REACT (Grant 756482).}
  \author{Ran Canetti}{Boston University, Boston, US \and \url{https://www.bu.edu/cs/profiles/ran-canetti/}}{canetti@bu.edu}{https://orcid.org/0000-0002-5479-7540}{}
  \author{Luowen Qian}{Boston University, Boston, US \and \url{https://cs-people.bu.edu/luowenq/}}{luowenq@bu.edu}{https://orcid.org/0000-0002-1112-8822}{}
  \authorrunning{Z.\ Brakerski, R.\ Canetti, and L.\ Qian}
  \Copyright{Zvika Brakerski, Ran Canetti, and Luowen Qian}
  \relatedversion{}
  \relatedversiondetails[linktext={arXiv:2209.04101}]{Full Version}{https://arxiv.org/abs/2209.04101}
  \funding{\emph{Ran Canetti} and \emph{Luowen Qian}: Supported by DARPA under Agreement No.\ HR00112020023.}
  \acknowledgements{\acks}
  
\EventEditors{Yael Tauman Kalai}
\EventNoEds{1}
\EventLongTitle{14th Innovations in Theoretical Computer Science Conference (ITCS 2023)}
\EventShortTitle{ITCS 2023}
\EventAcronym{ITCS}
\EventYear{2023}
\EventDate{January 10--13, 2023}
\EventLocation{MIT, Cambridge, Massachusetts, USA}
\EventLogo{}
\SeriesVolume{251}
\ArticleNo{35}
\else
  \newcommand{\email}[1]{\href{mailto:#1}{\texttt{#1}}}
  \author{
    \ifanon\else
      Zvika Brakerski\thanks{Weizmann Institute of Science, Israel,  \email{zvika.brakerski@weizmann.ac.il}. Supported by the Israel Science Foundation (Grant No.\ 3426/21), and by the European Union Horizon 2020 Research and Innovation Program via ERC Project REACT (Grant 756482).}
      \and
      Ran Canetti\thanks{Boston University, US. \email{canetti@bu.edu}, \email{luowenq@bu.edu}. Supported by DARPA under Agreement No.\ HR00112020023.}
      \and
      Luowen Qian$^\dagger$
    \fi
  }
  \date{}
\fi

\begin{document}
\maketitle

\ifexabs\else
\begin{abstract}
In the classical model of computation,
it is well established that \emph{one-way functions} (OWF) are minimal for computational cryptography:  They are essential for almost any cryptographic application that cannot be realized with respect to computationally unbounded adversaries. In the quantum setting, however, OWFs appear not to be essential (Kretschmer 2021; Ananth et al., Morimae and Yamakawa 2022), and the question of whether such a minimal primitive exists remains open. 

We consider EFI pairs --- efficiently samplable, statistically far but computationally indistinguishable pairs of (mixed) quantum states.
Building on the work of Yan (2022), which  shows equivalence between EFI pairs and statistical commitment schemes, we show that EFI pairs are necessary for a large class of quantum-cryptographic applications.  Specifically,  we construct EFI pairs from minimalistic versions of commitments schemes, oblivious transfer,  
 and general secure multiparty computation, as well as from   $\cQCZK$ proofs from essentially any non-trivial language.  We also construct  quantum computational zero knowledge ($\cQCZK$)  proofs for all of $\cQIP$ from any EFI pair.

This suggests that, for much of quantum cryptography,  EFI pairs play a similar role to that played by OWFs in the classical setting: they are simple to describe, essential, and also serve as a linchpin for demonstrating equivalence between primitives.
\end{abstract}

\ifintroonly
  \let\subsection\section
\fi
\section{Introduction}
\fi

One of the most fundamental achievements of  cryptography has been the conceptualization and eventual  formalization of the forms of computational hardness that are needed for obtaining prevalent cryptographic tasks. Notions such as one-way functions \cite{DH76,Yao82} (capturing functions that can be computed efficiently but are hard  to meaningfully invert)  and pseudorandom generators \cite{Shamir83,Yao82,BM84}  (capturing the ability to efficiently expand short random strings into longer strings that are hard to distinguish from fully random)   became foundational pillars for the design and reduction-based analysis of cryptographic schemes that are only ``computationally secure'' (that is,  secure only against computationally-bounded attacks). Furthermore, 
 the celebrated equivalence between the two notions \cite{BM84,GoldreichL89,HILL} has cemented the combined concept as the ``foundational computational hardness for cryptography'': One that is essential for realizing almost any cryptographic task that requires computational 
 hardness, and  at the same time suffices for realizing  a large class of cryptographic tasks. %

However,  in the quantum setting, where  parties  can generate, process, and communicate quantum information, the lay of the land of computational hardness turns out to be  different. First, quintessential tasks such as key-exchange with only public communication, which classically can only be computationally secure, can be obtained without any need for computational hardness \cite{bennett1984quantum,Ren08}. 
Furthermore, quantum protocols can use (quantum-hard)  one way functions to obtain tasks that are provably unobtainable from one way functions alone in the classical setting, at least in a relativizing manner.
These include non-interactive commitments with either statistical  hiding or statistical binding property \cite{MahmoodyP12,YWLQ15,BitanskyB21}, and oblivious transfer \cite{IR89,CrepeauLS01,BartusekCKM21a,GLSV21}.   

Even further,  it has been recently shown how to obtain commitments, oblivious transfer and general multiparty computation from a  form    of computational hardness that appears to be  
``purely quantum'',  in the sense that it  does not appear to imply one way functions (or any equivalent formulation of computational hardness)   ---  not even ones against \emph{classical}  attackers \cite{AQY21,morimae2021quantum}.
Superficially, this new form of computational hardness, called {\em  Pseudorandom States (PRS),} is a straightforward  generalization of pseudorandom generators:  it postulates the ability to efficiently generate quantum states that are hard to distinguish from a Haar-random state even when given multiple instances \cite{ji2018pseudorandom,BS19,BrakerskiS20}. However this apparent similarity  is deceiving;  indeed, there are no relativizing constructions of one way functions from PRS  \cite{Kretschmer21}.
Also, in spite of initial attempts \cite{morimae2021quantum}, the celebrated classical equivalence between one-wayness and pseudorandomness does not appear to naturally generalize to the quantum setting, at least not with respect to PRS.
Still, we do  not know whether PRS are \emph{essential}  for realizing any of the above cryptographic primitives.

This leaves quantum cryptography devoid of a convenient form of ``foundational computational hardness'',  namely a form of computational hardness that is both necessary for any meaningful computational security, and sufficient for realizing a large class of  tasks.

\paragraph*{Our contributions.}
We formulate a relatively \emph{simple and natural}  primitive and show that its existence is both necessary and sufficient for a significant class of cryptographic applications in a quantum-enabled computational model. While many of these implications are either known or easily derived from known results, we hope that the proposed framing,  along with the new implications, will help in understanding the computational foundations of quantum cryptography.

The proposed primitive draws from a classical primitive  considered by Goldreich \cite{Gol90},  as well as from the notion of canonical quantum commitments proposed by  Yan \cite{Yan22}.  Goldreich's primitive is aimed at capturing non-trivial and ``cryptographically useful'' computational indistinguishability:

\begin{definition}[EFID pairs \cite{Gol90}]
	An  EFID pair is a pair of efficient (classical) sampling algorithms such that their output distributions are statistically far but computationally indistinguishable.
\end{definition}

Goldreich's work leverages the 
result of Impagliazzo, Levin, and Luby~\cite{ILL89} to show that EFID pairs
 exist if and only if (classical) pseudorandom generators exist. 
This, together with what we know about pseudorandom generators, means the existence of a  classical protocol for almost any  cryptographic task that requires computational hardness implies existence of EFID pairs,
and furthermore that EFID pairs suffice for realizing a large class of tasks.
We consider a natural quantum analogue of EFID pairs:
\begin{definition}[EFI  pairs,  informal]
	An EFI pair is a pair of efficient quantum algorithms whose output \emph{states} are statistically far but computationally indistinguishable.
\end{definition}

Clearly,  any quantum-hard EFID pair is also an  EFI pair. 
On the other hand, it is unknown whether existence of EFI pairs  implies the existence of  (classical) pseudorandom generators or one-way functions.
Indeed, the implication is false in the relativizing setting (see  \cite{Kretschmer21}, combined with \cite[Theorem 4.1]{AQY21}). 

A first indication that EFI pairs are central to quantum  cryptography is the observation that they are essentially equivalent to the statically binding variant of canonical form quantum commitments  \cite{Yan22}\footnote{ Indeed, Yan~\cite{Yan22}  suggests studying the connections between statistically binding canonical-form commitment schemes and other cryptographic primitives. This work follows the same path, while distilling EFI as the notion of interest. See more details in \Cref{sec:tech-overview}.  %
}.
Building on this initial connection, we demonstrate that the existence of EFI pairs is essential for the existence of {\em  any} commitment scheme,  oblivious transfer protocols, non-trivial  multi-party computation protocols,  and  zero-knowledge proofs for non-trivial languages. Furthermore, for each one of these primitives,  we use EFI pairs as a tool for demonstrating that existence of protocols for a minimal version of the primitive implies existence of protocols for a full-fledged version of that primitive.
Informally:

\begin{itemize}

	\item\textbf{Quantum commitment schemes.}
	A commitment scheme is a cryptographic protocol where a committer commits to a hidden bit so that it can be later revealed but not modified.
	As mentioned,  EFI pairs  can be readily used to build (cannonical form) statistically binding non-interactive quantum commitments, which subsequently imply statistically hiding commitments \cite{Yan22}.
 
    We construct EFI pairs from any  plain ``semi-honest'' interactive commitment scheme,  namely an interactive commitment scheme that is (computationally) binding and hiding as long as both parties are honest during the commitment phase.  (For commitment schemes that are either statistically hiding or statistically binding, this implication is essentially shown in \cite{Yan22}. We extend this results to any commitment.)
    
	\item \textbf{Quantum oblivious transfer.}
	An oblivious transfer scheme is a cryptographic protocol where a sender makes two  bits available to a receiver  in a way that enables the receiver to  obtain exactly one of them, without the sender learning which bit it obtained. %
	Fully secure  (namely,  simulation-secure  against adversaries that deviate from the protocol) quantum oblivious transfer  is known to be constructible from quantum statistically binding commitments \cite{BartusekCKM21a,AQY21} (and so also from EFI pairs).
 
 We show how to construct EFI pairs from any \emph{semi-honest OT}  protocol, namely any OT protocol that is only guaranteed to be secure when both parties follow the protocol instructions without abort and up to purifications (namely, without tracing out any register used by each party).
	\item \textbf{Quantum secure multiparty computation  protocols.}
	A multiparty secure computation protocol is a protocol where participating parties jointly compute the output of a function of their secret inputs without revealing anything but the function value.
	Known constructions of general MPC from any statistically binding quantum commitment \cite{AQY21}  imply that EFI pairs can also be used to perform secure evaluation of any functionality. %
 
	We show how any protocol for securely evaluating any non-trivial classical finite functionality (namely a function with an insecure minor  as in \cite{BeimelMM99}), even in the semi-honest model, implies the existence of EFI pairs.
	\item \textbf{Quantum computational zero-knowledge ($\cQCZK$) proofs.}
	Finally, a computational zero-knowledge proof is an interactive proof system where any malicious verifier cannot learn anything beyond the fact that the statement is true, in the sense that their view could be efficiently simulated  given only the public instance.
	
    By observing that zero knowledge is a special case of secure two-party computation (2PC), we have that if EFI pairs exist then $\cQCZK = \cQIP$, and any language in $\cQMA$ also admits $\cQCZK$ proofs with negligible soundness and an efficient prover that uses only a single copy of the witness.
	Conversely, we build on results from \cite{YWLQ15} to construct EFI pairs from any \emph{honest verifier} $\cQCZK$ proof ($\cQCZKHV$) for any language that is hard on average for $\cBQP$.
	($\cQCZKHV$ is a relaxation of $\cQCZK$ where  zero knowledge property is guaranteed to hold only against purified verifiers with abort, rather than arbitrary polytime verifiers.)
\end{itemize}

Furthermore, all these equivalences relativize\ifexabs\else\footnote{
  We use the term \emph{relativizing} to denote that the construction works even in the presence of (quantum) oracles.
  In particular, this allows the construction to invoke the next message function of a protocol as a black box or even run its purification.
	An example of a relativizing implication is Watrous's  construction of a  $\cQSZK$ protocol from  any honest-verifier $\cQSZKHV$ protocol \cite{Wat02}.

  We  reserve the use of the term  \emph{black-box  reductions} to denote reductions which are black  box in a primitive,  namely reductions that use   oracle access to the ideal functionality describing the primitive, irrespective of any particular implementation (as e.g.\ in the work of Kilian~\cite{Kil91}).
}\fi.
Thus, Kretschmer's oracle separation \cite{Kretschmer21} immediately generalizes to show that none of the objects constructible from EFI pairs (or pseudorandom states) imply the existence of one-way functions (post-quantum or not) in a relativizing way.

\ifexabs\else
\subsection{Our techniques}
\label{sec:tech-overview}

This section presents an overview of the proofs for our results.

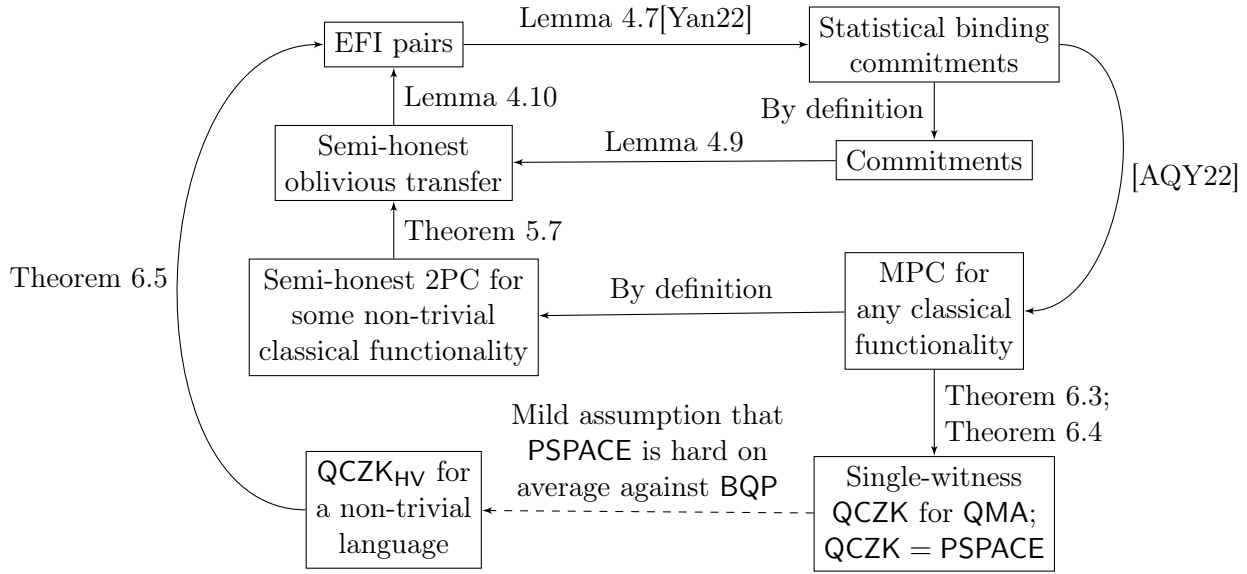
\begin{figure}
  \centering
  \tikzstyle{block} = [rectangle, draw, text centered, align=center]
\tikzstyle{implies} = [draw, -latex']
\tikzstyle{equiv} = [draw, latex'-latex']

\begin{tikzpicture}
\node (efi) [block] {EFI pairs};
\node (statcom) [block, right=12em of efi] {Statistical binding\\commitments};
\node (com) [block, below=2em of statcom] {Commitments};
\node (shot) [block, below=2em of efi] {Semi-honest\\oblivious transfer};

\path [implies] (efi) -- node [above] {\ifintroonly\else\Cref{lem:qg2commitments}\fi \cite{Yan22}} (statcom);
\path [implies] (statcom) -- node [left] {By definition} (com);
\path [implies] (com) -- node [above] {\ifintroonly\else\Cref{lem:comm2sot}\fi} (shot);
\path [implies] (shot) -- node [right] {\ifintroonly\else\Cref{lem:sot2qg}\fi} (efi);

\node (mpc) [block, below=2.5em of com] {MPC for\\any classical\\functionality};
\node (sh2pc) [block, below=2em of shot] {Semi-honest 2PC for\\some non-trivial\\classical functionality};

\path [implies] (statcom) to[out=0,in=0] node [right] {\cite{AQY21}} (mpc);
\path [implies] (mpc) -- node [above] {By definition} (sh2pc);
\path [implies] (sh2pc) -- node [right] {\ifintroonly\else\Cref{thm:sh2pc2efi}\fi} (shot);

\node (qczkstrong) [block, below=3em of mpc] {Single-witness\\$\cQCZK$ for $\cQMA$;\\$\cQCZK = \cPSPACE$};
\node (qczkweak) [block, below=2.6em of sh2pc] {$\cQCZKHV$ for\\a non-trivial\\language};

\path [implies] (mpc) -- node [right, align=left] {\ifintroonly\else\Cref{thm:efi2qczk-efficient};\\\Cref{thm:qczk-qip}\fi} (qczkstrong);
\path [implies, dashed] (qczkstrong) -- node [above, align=center] {Mild assumption that\\$\cPSPACE$ is hard on\\average against $\cBQP$} (qczkweak);
\path [implies] (qczkweak) to[out=180,in=180] node [left, align=left] {\ifintroonly\else\Cref{thm:hardness2qci}\fi}(efi);
\end{tikzpicture}
  \caption{
    Our results for EFI pairs illustrated.
    We give a more detailed overview of these implications in \Cref{sec:tech-overview}.
  }
  \label{fig:results}
\end{figure}

\paragraph*{EFI pairs and statistical commitments.}
In the classical setting,   building a commitment scheme from  EFID pairs would naturally go via  Goldreich's transformation  to a PRG,  and then use,  say,
Naor's commitment \cite{Naor91}. However, it is not clear how this transformation could be generalized to the quantum setting. In particular,
Goldreich's proof crucially relies on the fact that for a $\cBPP$ (randomized) algorithm, it is possible to separate the randomness from the rest of the computation --- or even arbitrarily \emph{program} the randomness.  
 Such techniques cannot work for quantum algorithms, as also observed by the recent work of Aaronson, Ingram, and Kretschmer~\cite{AIK21} comparing the complexity classes $\cBPP$ and $\cBQP$.

Still, as noticed several times in the literature, 
 EFI  pairs give quantum  commitments in a rather direct way.
To sketch this basic construction, we first recall the syntax of a canonical form quantum commitment scheme:\footnote{
  This (non-interactive) canonical form of quantum bit commitment schemes was first introduced by the work of Yan et al.~\cite{YWLQ15}, but the idea dates back to the work of Chailloux, Kerenidis, and Rosgen~\cite{ChaillouxKR11}.
  Subsequently, it has been shown that canonical form commitments are as useful as traditional bit commitments by Yan et al.~\cite{YWLQ15,FUYZ20,Yan22}
  Its connection to (classical) EFID pairs was observed by Yan in a 2022 revision of their work~\cite{Yan22}.
}
\begin{itemize}
	\item To commit to $b$, the committer efficiently generates a bipartite (i.e.,  two-register)  state $\ket{\psi_b}_{\mathsf{CR}}$, and sends the commitment register $\mathsf C$.
	\item To open, simply reveal the other register $\mathsf R$ and $b$, and the receiver can perform a rank-$1$ projection onto the corresponding state $\ket{\psi_b}_{\mathsf{CR}}$ (or equivalently, uncompute the state generation unitary for committing to $b$ and check if we get back all zeroes) to check whether to accept the commitment.
\end{itemize}

It is possible to view the purified generation of the EFI pair as a  canonical form commitment where the output corresponds to the commitment register $\mathsf C$ and the purification corresponds to the opening register $\mathsf R$. When viewed this way,  the statistical distance guarantee of the EFI  pair translates to  statistical binding property,  whereas  the computational indistinguishability of the EFI  pair translates to the computational hiding  property of the commitment. 
Thus an EFI pair is essentially a statistically binding canonical form commitment.  This  observation  (which is implicit in \cite{Yan22}) is indeed the starting point of our work.

Furthermore, the round collapsing theorem in \cite{Yan22} shows that any quantum commitment can be compiled into the canonical form while preserving the hiding and binding properties.
Since statistically binding commitments and statistically hiding commitments are equivalent \cite{CrepeauLS01,Yan22}, we can construct EFI pairs from either one.

Showing how to construct EFI pairs from commitment schemes that are neither statistically binding nor statistically hiding appears more challenging.
One may hope to somehow construct a candidate EFI pair of states, and prove computational indistinguishability from computational security of the commitment, and statistical distance via an inefficient attack on the commitment scheme.
However, it is not clear how to transform an inefficient attack against the binding property into a distinguishing attack as needed for EFI pairs.
(Classically, this  part can be done using one-way functions,  but, as argued above, these techniques do not have natural quantum analogues.) %

We get around this difficulty by going through oblivious transfer, where the security for both ends can be naturally viewed as distinguishing tasks, and is thus more amenable to constructing EFIs even without  statistical security. 
In fact, we observe that semi-honest OT suffices.  Let us elaborate.

\paragraph*{EFI pairs from semi-honest OT.}
We first define a new  notion of quantum oblivious transfer,  which considers only ``purification attacks''. %
(This notion can be viewed as a quantum analogue of the classical ``non-erasing, honest-but-curious'' attacks.)
Furthermore, the semi-honest adversaries must run the purified protocol till the end, i.e.\ are disallowed to abort.
(Classically, this does not matter.)

We first describe constructing EFI from semi-honest OT, which is more straightforward.
Consider the adversarial view in the following two  executions:
\begin{itemize}
  \item
    Here the sender is honest and the receiver is semi-honest. The  sender chooses the bits uniformly at random and the semi-honest receiver chooses the choice bit uniformly at random.
    (By correctness, the receiver is always able to recover one bit specified by the choice bit with certainty.
    The task for the semi-honest receiver is to extract the sender's other bit from his purified view.)
  \item
    Here the receiver is honest and the sender is semi-honest. The receiver chooses the choice bit uniformly at random and the semi-honest sender prepares two equal superposition states $\ket+^{\otimes 2}$ for her input bits.
    (Here the task for the semi-honest sender is to extract receiver's choice bit from her purified view.)
\end{itemize}
It is easy to see that these views can be computed efficiently since the OT protocol is efficient; and by the semi-honest security of OT, both tasks should be impossible for efficient algorithms.
At the same time, the impossibility of Chailloux, Gutoski, and Sikora~\cite{CGS16}  states that for every OT protocol, one of these two tasks can be accomplished with success probability $\ge\frac23$ inefficiently.

Now the construction of EFI pairs follows  by simply reinterpreting these bit extraction tasks as a distinguishing task.
In particular, for $b = 0, 1$, the $b$-th state   is simply the concatenation of the  two above views, conditioned on the correct answers being $b$ in both executions.
In other words, the $b$-th state consists of the semi-honest receiver's view of the first execution when the sender's other bit is $b$, followed by the semi-honest sender's view of the second execution when the receiver's choice bit is $b$.

\paragraph*{Semi-honest OT from commitments.}
Many prior works have already studied constructing quantum oblivious transfer from commitments \cite{CrepeauK88,BennettBCS91,Crepeau93,CrepeauLS01,FUYZ20,Yan22}.
However, they all start with a commitment with some statistical security guarantee -- either statistically binding or statistically hiding. On the other hand, we want to start from an arbitrary commitment scheme (which may be computationally binding and computationally hiding).
While these constructions could probably still carry over, since here we are aiming for a much weaker security, we instead give a much simplified protocol with a self-contained description.
Let us begin by considering the simplest (almost trivial) quantum oblivious transfer protocol inspired by Cr\'epeau and Kilian \cite{CrepeauK88}, which is only secure if both parties are completely honest during the protocol:
\begin{enumerate}
	\item The sender on input two message bits $b_0, b_1$, sends two qubits $\ket{b_0} \otimes H\ket{b_1}$, where the first qubit encodes the first message bit in the standard basis and the second qubit encodes the second message bit in the Hadamard basis.
	\item The receiver measures both qubits in the standard basis to recover $b_0$ and a random bit $b_1'$, or in the Hadamard basis to recover a random bit $b_0'$ and $b_1$.
\end{enumerate}
Security is straightforward: the sender gets no information at all, and the receiver destroys the information about the other bit by measuring it in an incompatible basis.

However, this is obviously not semi-honest secure as a purified receiver could simply uncompute a purified measurement to recover the other bit as well.
Indeed, a better way to ``erase'' information in the semi-honest model is to simply send it to the other party: in this case, the receiver sends a copy of these two measured (classical) bits to the sender.
Since the semi-honest receiver's purified view at the end has the sender's private registers traced out, doing this ensures that these measured bits do indeed collapse even for the purified view.

On the other hand, since these two bits contains information about the receiver's choice bit, we cannot simply send it in the clear as now the sender can break.
One simple fix is to have the receiver instead commit to the two measured bits to the sender.
This strategy is also often employed when designing maliciously secure OT from commitments \cite{BennettBCS91,Crepeau93,CrepeauLS01}.

It is easy to see that this committed-measurement OT remains secure against semi-honest sender by hiding of the commitment, and security against semi-honest receiver remains to be seen.
If the commitment is statistically binding, then it can be seen that the collapse still occurs; but the computational binding case seems less clear.
Fortunately, this can be overcome through Yan's computational collapse theorem \cite{Yan22}, which on a high level states that a canonical form computationally binding commitment scheme computationally ``collapses'' the commited qubit (given that the commit phase was performed semi-honestly), even if the commitment is never opened.
(Their theorem is established via reducing a collapsing distinguishing adversary to an adversary that breaks computational binding of the canonical form commitment.)
This completes the argument.

\paragraph*{Multiparty secure computations for classical functionalities.}
Using a known sequence of transformations outlined in  existing works \cite{AQY21} (which builds on existing works including but not limited to \cite{BennettBCS91,BartusekCKM21a,GLSV21}), it is already known how construct, given any statistically binding quantum commitment scheme (and hence also given any EFI pair), multi-party protocols that securely evaluate any classical function with any number of faults.
Furthermore, these protocols provide statistical security guarantees against at least one of the parties.
As an aside, via known results, these protocols can further be used to construct two-party secure computations for general quantum functions (or channels) where only one party obtains output \cite{DupuisNS12}, and in addition,  reactive (meaning stateful and interactive) classical functionalities \cite{CrepeauGT95,IshaiPS08}.

For the converse direction, we can now use  the powerful equivalence established for oblivious transfer above, and  simply invoke the classical equivalence of Beimel, Malkin, and Micali~\cite{BeimelMM99} to complete the proof.
While the \cite{BeimelMM99}  proof contains parts which do not naturally generalize to the quantum setting, the only thing we need from that proof is the reduction from semi-honest OT to semi-honest 2PC for any non-trivial classical functionality (i.e.\ if it contains an insecure minor), and this construction is black-box and hence extends to our setting.
While the semi-honest models are slightly different, we verify that their semi-honest reduction also works for our model.
Once we have semi-honest OT, we get EFI pairs by the equivalence above.

\paragraph*{Zero knowledge proofs from EFI pairs.}
We now turn to establishing an equivalence between EFI pairs and non-trivial quantum computational zero knowledge ($\cQCZK$) proofs. %
We  first consider the task of constructing $\cQCZK$ protocols for $\cQMA$ from EFI pairs.
Here the commit-and-open $\cQCZK$ protocol  by Broadbent and Grilo~\cite{BroadbentG20} can be readily instantiated by any quantum commitment. However, this protocol uses sequential repetition, and as a consequence, requires multiple copies of the quantum witness to achieve negligible soundness.
As also proposed by Broadbent et al.~\cite{BJSW20}, this limitation can be avoided via performing 2PC for quantum (CPTP) functionalities, which can be constructed from OT \cite{DupuisNS12} and thus EFI pairs.

We now move on to general $\cQCZK$ proofs without any constraint on prover efficiency and show how to use EFI pairs to construct $\cQCZK$ proofs for all of $\cQIP$.
Before presenting our protocol, let us recall the celebrated construction of Ben-Or et al.~\cite{Ben-OrGGHKMR88} that transforms any (without loss of generality, Arthur-Merlin or public-coin) $\cIP$ protocol into a $\cCZK$ protocol.
In the transformed protocol, the parties first run the original public-coin protocol, where the prover only sends (statistically binding) commitments to its messages. Next, the parties engage in a zero-knowledge protocol where the instance consists of the transcript so far, and the language accepts a transcript if there exist valid openings to all the prover commitments that would have caused the original verifier to accept. 

Since $\cQIP = \cIP$ \cite{Sha92,JJUW11}, it is natural to consider extending  the \cite{Ben-OrGGHKMR88} construction to our setting. 
However, direct extension hits a roadblock: the statement that needs to be proven in zero knowledge is now a \emph{quantum statement} involving the commitment states when instantiated with quantum commitments.
This is not a context that is traditionally considered by the zero knowledge literature. (Indeed, recall that even zero knowledge proofs for $\cQMA$ still consider classical statements.)
Even if we attempt to mimic a zero knowledge proof via statistical 2PC, we soon encounter another issue: how should the two parties agree on the quantum statement that is being proven?
Sure, we could make the verifier send the quantum state in the statement to 2PC, but a malicious verifier could refuse to provide the correct state.
This becomes an issue as the verifier might be able to manipulate the commitment message so that checking the validity of the commitments itself might reveal non-trivial information about the committed bit.

We thus take a different path: we have the prover and the verifier engage in a secure evaluation of the following 
reactive functionality (which also can be constructed from OT \cite{CrepeauGT95,IshaiPS08}). The verifier inputs its random challenges in the underlying interactive proof, and the functionality uses these challenges to  play the verifier role in an interactive proof with the external prover. Finally  the functionality outputs the acceptance bit to the external verifier.   Both soundness and zero knowledge follow from the security of the MPC.

\paragraph*{EFI pairs from non-trivial $\cQCZKHV$.}
Finally,  we show how to construct EFI  pairs from any   $\cQCZK$ proof for any language that is hard on average against $\cBQP$.
Note that the computational indistinguishability given by the $\cQCZK$ security does not give EFI pairs immediately as it might be possible to generate the hoenst view efficiently.

One possible approach might be to try to extend the classical result of Ostrovsky and Widgerson~\cite{OW93} to our setting.
However, they use the non-existence of one way functions to build universal extrapolators that efficiently turn simulators into cheating provers, and it is not clear how to use the non-existence of EFI pairs to construct quantum universal extrapolators.

We instead turn to the works of Ong and Vadhan~\cite{Vadhan06,OngV08}, showing an equivalence between instance-dependent commitments and $\cCZK$, that is a language admits an instance-dependent commitment (a commitment, parameterized by an instance $x$, whose computational hiding and statistical binding properties only hold if $x$ is in or not in the language, respectively) if and only if it admits a $\cCZK$ proof.
We note that if a hard-on-average language admits an instance-dependent commitment, then this commitment is essentially a full-fledged commitment, thus implying the existence of EFI pairs.
Therefore, it remains to extend the equivalence to the quantum setting, i.e.\ we wish to establish that any language admits a $\cQCZK$ proof if and only if it admits instance-dependent quantum commitments.

The bad news is that going from $\cCZK$ to instance-dependent commitments again involve going through instance-dependent one-way (universal hash) functions.
However, we note that the mixed states considered by Watrous~\cite{Wat02} for handling $\cQSZK$ readily gives an instance-dependent mixed state for $\cQCZK$ protocols: a weak variant of EFI states that is only required to satisfy either statistical farness or computational indistinguishability if $x$ is in or not in the language, respectively.
It can be seen that the transformations described before also readily extends to the instance-dependent setting, and thus this gives an instance-dependent commitment.
Indeed, this transformation from $\cQCZK$ to instance-dependent quantum commitments has  been observed by the work of Yan et al.~\cite{YWLQ15}.
We then conclude that if $L$ is hard on average for $\cBQP$, then this instance-dependent mixed state averaged over the hard distribution immediately gives an EFI pair.

Upon the completion of this work, we discovered a result by Chailloux, Kerenidis, and Rosgen~\cite{ChaillouxKR11} that is similar to this part with very similar proof techniques.
However, their separation is between $\cQIP = \cPSPACE$ and $\cQMA$, which is technically incomparable with our separation between $\cQCZK$ and $\cBQP$.
Furthermore, they consider worst case hardness instead of average case hardness here, and thus only getting quantum auxiliary-input EFI.
In our case, this difference is rather minor and the results can translate back and forth\ifintroonly\else{} (see \Cref{thm:aiefi})\fi; and in their case, it is not clear how to get standard EFI pairs from any notion of average case hardness of $\cQIP$ against $\cQMA$.\footnote{
  For readers that are familiar with \cite{ChaillouxKR11}, even if we assume $\cQIP$ is hard on average against $\cQMA$ for $\cBQP$ samplable distributions, this still does not suffice for getting EFI pairs without quantum auxiliary input.
  The reason is that the quantum auxiliary input needs to specify the state that witnesses the diamond norm of the two channels, and it is not clear how this state could be prepared efficiently.
}

\subsection{Discussions and open questions}
\label{sec:discussions}

We now give a few open questions in this direction, organized into three categories.
To keep the discussion succinct, we point the readers to the references for details of the terminologies.

\paragraph*{EFI and quantum complexity.}

One way functions (and equivalently pseudorandom generators and classical EFID pairs) have been one of the central objects in complexity theory \cite{AB09}.
Since EFI pairs  are both essential and sufficient for much of quantum cryptography, and furthermore are very simple to describe, it is natural to ask whether EFI pairs could also be a useful object to study from the complexity point of view.
Note that the computational hardness underlying EFI pairs, which is the quantum state distinguishability problem, seems especially relevant to the study of the complexity of quantum states and transformations \cite{Aar16}.

One very important question, we think, is whether there is any barrier for proving the existence of EFI pairs.
In the classical setting,  existence of one-way functions implies $\cP \neq \cNP$, but is there any barrier for establishing the existence of quantum EFI pairs?
EFI pairs would immediately imply a quantum circuit lower bound for an explicit two-outcome measurement, but is there any reason to believe that such a lower bound would be hard to establish?
	
	For a more concrete example, is $\cP$ vs $\cPSPACE$ a (classical) barrier for the existence of EFI pairs?
	In other words, does the existence of EFI pairs separate $\cBQP$ from $\cPSPACE$?
	We know that the existence of pseudorandom states do separate $\cBQP$ from $\mathsf{PP} = \mathsf{PostBQP}$ \cite{Kretschmer21}, but nothing is known for EFI pairs.
	One way to achieve this could be to demonstrate a way to synthesize the Helstrom measurement given a $\cPSPACE$ oracle, which is closely related to the unitary synthesis problem \cite{AK07}: in particular, if the unitary synthesis problem could be done by a $\cPSPACE$ oracle, then the existence of EFI pairs would separate $\cPSPACE$ from $\cBQP$.

 \paragraph*{Hardness amplification for EFI.} Another intriguing challenge is hardness amplification for EFI.  Is it possible to construct full-fledged  EFI pairs from  weaker veriants where either the computational distance is non-negligible, or the statistical distance is bounded away from 1,  or both?  What is the minimal initial gap between the statistical and computational  distances that still allows amplification?

\paragraph*{Candidate EFI?}
Given the oracle separation, are there any concrete candidate assumptions that imply quantum EFI, with formal evidence that it does not imply one-way functions?

Two natural candidates come from pseudorandom state candidates.
One possible approach is to assume that ``sufficiently large'' quantum random circuits are pseudorandom unitaries.
These random quantum circuits are already being investigated with motivations like quantum supremacy \cite{AA13,AaronsonC17} and the theory of black holes \cite{BS18,BCHKP21,Haferkamp2022}.
One could hope that pursuing this direction could ultimately lead to useful quantum cryptography that could be implemented on near-term quantum devices.
The other possible approach is to consider the pseudorandom states proposed by Bouland, Fefferman, and Vazirani~\cite{BoulandFV20} from the physical description of wormholes.
Intuitively, such a construction could be secure based on the physical belief that a wormhole is highly ``scrambling''.
In that paper, they also prove that their construction is a secure pseudorandom state generator if the evolution unitary is a black-box Haar random unitary.

Another candidate is proposed by Kawachi, Koshiba, Nishimura, and Yamakami~\cite{Kawachi2012}.
On a high level, they consider computational indistinguishability between two types of random coset states, and show that it is at least as hard as the graph automorphism problem.

\paragraph*{EFI and quantum cryptography.}
The importance of one way functions in classical cryptography cannot be overstated: virtually any non-trivial computational cryptography
(those that cannot be realized with respect to computationally unbounded adversaries)
implies the existence of one way functions classically.
Yet it still remains to be seen how much of quantum cryptography is related to EFI.
For instance:
\begin{description}
	\item[Quantum pseudorandomness.]
	The celebrated result of Goldreich~\cite{Gol90} shows existence of classical EFI pairs imply pseudorandom generators and subsequently pseudorandom functions.
	While we know how to construct various quantum cryptography from quantum EFI pairs, the way we do it completely avoids the need to construct quantum pseudorandomness.
	Nevertheless, given the many applications of both classical and quantum pseudorandomness, an important question is whether it is possible to construct quantum pseudorandomness (pseudorandom states, unitaries, or any other meaningful pseudorandom objects) from EFI pairs.
	 \item[Quantum unforgeability.]
 Cryptographic primitives such as  digital signatures,  message authentication codes, or quantum money  appear to inherently require some flavor of one-wayness,  in that a break involves solving a computationally hard search problem where solutions exist and  are efficiently verifiable given some additional secret information. (In the public key setting, where some classical verification key is made public, solutions are verifiable publicly.)
 Another related object already proposed previously is one-way state generators \cite{morimae2021quantum}.
 As discussed, this form of computational hardness appears very different than indistinguishability. Still,  can we show that existence of any one of these primitives implies existence of EFI  pairs?  Can we construct any of these primitives from EFI  pairs?

		  \item[Quantum zero knowledge arguments.]
	It is possible to extend our proof \ifintroonly\else(\Cref{thm:efi2qczk-efficient})\fi{} to show that if EFI pairs exist then we can give a $\cQSZK$ argument for any $\cQMA$ language with an efficient prover having a single copy of the witness.
	On the other hand, we are unable to show that if $L$ admits an argument then the instance-dependent commitment \ifintroonly\else\ifitcsproc\else{}constructed in \Cref{thm:qczkhv-qczk}\fi\fi{} is computationally binding for NO instances.
	On a high level, the difficulty is that known techniques in the statistical binding setting do not translate to the computational setting, since we do not have a hardness amplification procedure for computationally binding commitments: it is not clear whether parallel repetition of commitments decreases the computational binding error against malicious committers.
	We refer the readers to the related discussions in Yan's work \cite{Yan22} for more details.
\end{description}

\fi

\ifexabs\else
\ifintroonly\else
\section{Preliminaries}

\subsection{Quantum information}

We refer the reader to~\cite{nielsen_chuang_2010} for a comprehensive reference on the basics of quantum information and quantum computation. We use standard Dirac notation for quantum states. 

We recall the notion of density matrices, which are PSD trace-$1$ matrices that represent the complete characterization of a state of a quantum system. The state of a system can be ``pure'', i.e.\ in the form of a state $\ket{\psi}$, in which case the density matrix is $\ketbra{\psi}{\psi}$ (i.e.\ of rank $1$), or ``mixed'' which corresponds to a distribution over pure states, and is represented by a density matrix of rank $> 1$. Two quantum states are identical if and only if their density matrices are equal, and the distance between quantum states is also expressed as a function of their density matrices, as explained below. 

We recall that quantum operations can always be expressed as unitary operators on \emph{some} quantum system, and they act on the density matrix of this quantum system via conjugation. We may sometimes refer to a quantum operation that uses some auxiliary registers, or that removes (``traces out'') registers during the computation. Such a general quantum operation is known as a \emph{quantum channel}. (Mathematically, a quantum channel can be expressed as a completely-positive trace-preserving (CPTP) map on the space of density matrices, but this formulation will not be required for our purposes.)

We use $\mathcal D(\mathcal{H})$ to denote the set of density matrices on a Hilbert space $\mathcal{H}$.
Let $\rho,\sigma \in \mathcal D(\mathcal{H})$ be density matrices. We write $\TraceDist(\rho,\sigma)$ to denote the trace distance between them, i.e.,
\[
\TraceDist(\rho,\sigma) = \frac{1}{2} \norm{\rho - \sigma}_1
\]
where $\norm{X}_1 = \Tr(\sqrt{X^\dagger X})$ denotes the trace norm.
We also use $F(\rho, \sigma) = \paren{\Tr\mparen{\sqrt{\sqrt\rho\sigma\sqrt\rho}}}^2$ to denote the fidelity of $\rho$ and $\sigma$.

\begin{fact}
  \label{fact:tdnf}
  For any two mixed states $\rho, \sigma$, $\mparen{F(\rho, \sigma)}^2 + \mparen{\TraceDist(\rho, \sigma)}^2 \le 1$.
\end{fact}

\begin{theorem}[{Holevo--Helstrom~\cite{Holevo73,Helstrom69}}]
  \label{thm:holevo-helstrom}
  The best success probability to (inefficiently) distinguish two mixed states $\rho, \sigma$ is given by $\frac12\paren{1 + \TraceDist(\rho, \sigma)}$.
  The measurement that achieves this success probability is called the \emph{Helstrom measurement}.
\end{theorem}
\begin{corollary}
  \label{cor:td-amplification}
  For any mixed states $\rho, \sigma$ and integer $n > 0$,
  \[ \TraceDist(\rho^{\otimes n}, \sigma^{\otimes n}) \ge 1 - \exp\mparen{-n\TraceDist(\rho, \sigma)/2}. \]
\end{corollary}
We recall that $\rho^{\otimes n}$ is the state containing $n$-copies of the state represented by $\rho$.
\begin{proof}
  We prove this by building a majority-vote distinguisher.
  Let $D$ be the distinguisher optimally distinguishing $\rho$ from $\sigma$ by Holevo--Helstrom theorem.
  We apply $D$ on each copy, and take the majority vote.
  Corollary follows by applying Hoeffding's inequality and Holevo--Helstrom again.
\end{proof}

\subsection{Quantum algorithms}

A quantum algorithm $A$ is a family of generalized quantum circuits $\{A_n\}_{n \in \mathbb N}$ over a discrete universal gate set (such as $\{ CNOT, H, T \}$). By generalized, we mean that such circuits can have a subset of input qubits that are designated to be initialized in the zero state, and a subset of output qubits that are designated to be traced out at the end of the computation. Thus a generalized quantum circuit $A_n$ corresponds to a \emph{quantum channel}.
When we write $A_n(\rho)$ for some density matrix $\rho$, we mean the output of the generalized circuit $A_n$ on input $\rho$. If we only take the quantum gates of $A_n$ and ignore the subset of input/output qubits that are initialized to zeroes/traced out, then we get the \emph{unitary part} of $A_n$, which corresponds to a unitary operator which we denote by $\hat{A}_n$. The \emph{size} of a generalized quantum circuit is the number of gates in it, plus the number of input and output qubits.

We say that $A = \{A_n\}_n$ is a \emph{quantum polynomial-time (QPT) algorithm} if there exists a polynomial $p$ such that the size of each circuit $A_n$ is at most $p(n)$. We furthermore say that $A$ is \emph{uniform} if there exists a deterministic polynomial-time Turing machine $M$ that on input $1^n$ outputs the description of $A_n$. 

We also define the notion of a \emph{non-uniform} QPT algorithm $A$ that consists of a family $\{(A_n,\rho_n) \}_n$ where $\{A_n\}_n$ is a polynomial-size family of circuits (not necessarily uniformly generated), and for each $n$ there is additionally a subset of input qubits of $A_n$ that are designated to be initialized with the density matrix $\rho_n$ of polynomial length. This is intended to model nonuniform quantum adversaries who may receive quantum states as advice.  %

The notation we use to describe the inputs/outputs of quantum algorithms will largely mimic what is used in the classical cryptography literature. For example, for a state generator algorithm $G$, we write $G_n(k)$ to denote running the generalized quantum circuit $G_n$ on input $\ketbra{k}{k}$, which outputs a state $\rho_k$. 

Ultimately, all inputs to a quantum circuit are density matrices. However, we mix-and-match between classical, pure state, and density matrix notation; for example, we may write $A_n(k,\ket{\theta},\rho)$ to denote running the circuit $A_n$ on input $\ketbra{k}{k} \otimes \ketbra{\theta}{\theta} \otimes \rho$. In general, we will not explain all the input and output sizes of every quantum circuit in excruciating detail; we will implicitly assume that a quantum circuit in question has the appropriate number of input and output qubits as required by context.

We assume that all  parties are  quantum algorithms with (noiseless) quantum communication.  Furthermore,  all algorithms run with the same  security parameter.

A function $f: \mathbb N \to \mathbb R^{\ge 0}$ is negligible, if for any polynomial $p$, $f(\lambda) \le 1/p(\lambda)$ for all sufficiently large $\lambda \in \mathbb N$.
Otherwise, we say it is noticeable, or equivalently when it is infinitely often at least $1/p(\lambda)$.

\begin{definition}[Computational indistinguishability]
  For two families of mixed states $\{\rho_\lambda\}_\lambda, \{\sigma_\lambda\}_\lambda$, we say that they are computationally indistinguishable (against $\cBQPqpoly$), if for any QPT algorithm $D$, there exists a negligible function $\varepsilon$ such that for any security parameter $\lambda$ and advice state $\alpha = (\alpha_\lambda)_\lambda$,
  \[ \abs{\Pr[D(\alpha_\lambda, \rho_\lambda) = 1] - \Pr[D(\alpha_\lambda, \sigma_\lambda) = 1]} \le \varepsilon(\lambda). \]
\end{definition}
Extending the standard cryptographic convention,the above definition considers  adversaries with non-uniform quantum advice.
The definition  can be adapted to the uniform setting by simply requiring fixing $\alpha = \bot$.
We note that the reductions shown in this work are all uniform, or ``advice preserving'': given an adversary with some advice, the generated adversary uses the same advice.

\section{EFI pairs of states}

We define the main object considered in this work, namely pairs of efficiently generatable mixed quantum states that are statistically far and yet computationally indistinguishable:
 
\begin{definition}
  We call $\xi = (\xi_{b, \lambda})$ a pair of EFI states if it satisfies the following criteria:
  \begin{enumerate}
    \item \textbf{Efficient generation}: There exists a uniform QPT quantum algorithm $A$ that on input $(1^\lambda, b)$ for some integer $\lambda$ and $b \in \{0, 1\}$, outputs the mixed state $\xi_{b, \lambda}$.
    \item \textbf{Statistically Far}: $\TraceDist\mparen{\xi_{0, \lambda}, \xi_{1, \lambda}}$ as a function of $\lambda$ is at least inverse polynomial.
    \item \textbf{Computational Indistinguishability}: $(\xi_{0, \lambda})_\lambda$ is computationally indistinguishable to $(\xi_{1, \lambda})_\lambda$.
  \end{enumerate}
\end{definition}

Here we require exact generation of the mixed state.
Since we only care about the existence of such object, this requirement does not make a difference.
In particular, if we can approximately synthesize a certain family of states with inverse-exponential fidelity, then taking the output of the circuit directly would also satisfy the requirement.

\section{Commitments and semi-honest oblivious transfer}

\subsection{Commitments}

In this work, we without loss of generality focus on the canonical form of quantum commitment schemes \cite[Definition 5]{Yan22}.
A commitment scheme consists of two phases.
In the commitment phase of a canonical commitment scheme, Alice (the committer) chooses a bit $b$, and runs a uniform QPT circuit $Q_{\lambda, b}$ on all zeroes, which outputs two registers $\mathsf C, \mathsf R$; she then proceeds to send the register $\mathsf C$ to Bob (the receiver).
Later in the reveal phase, Alice sends the other register $\mathsf R$ and the bit $b$ to Bob; Bob accepts the opening if he performs $Q_{\lambda, b}^\dagger$ on two registers and measures all zeroes in the computational basis.
We now recall the requirements on the commitment schemes, specialized to canonical forms for convenience.
It will be convenient to define the commitment message $\rho_{\lambda, b} := \Tr_{\mathsf{R}}(Q_{\lambda, b}\ketbra00Q_{\lambda, b}^\dagger)$.

\begin{definition}[Computational hiding]
  A commitment scheme satisfies computational hiding, if $\rho_{\lambda, 0}$ is computationally indistinguishable to $\rho_{\lambda, 1}$.
\end{definition}

We are going to consider a specific more restricted variant of statistical binding called honest binding.
We refer the readers to related works \cite{YWLQ15,FUYZ20,morimae2021quantum,Yan22} for a more thorough discussion on this variant (and how it is equivalent to statistical binding for canonical commitment schemes).

\begin{definition}[Honest binding]
  A canonical commitment scheme satisfies honest computational (resp.\ statistical) binding if for any auxiliary state $\ket\psi$ and any polynomial-time (resp.\ physically) realizable unitary $U$, we have that
  \[ \norm{\paren{Q_{\lambda, 1}\ketbra00_{\mathsf{CR}}Q_{\lambda, 1}^\dagger \otimes I_{\mathsf Z}} \paren{I_{\mathsf C} \otimes U_{\mathsf{RZ}}} \paren{Q_{\lambda, 0}\ket0_{\mathsf{CR}} \otimes \ket\psi_{\mathsf Z}}}_2 \]
  is negligible.
\end{definition}

\subsection{Oblivious transfer and semi-honest adversaries}

In an oblivious transfer protocol, Bob (the sender) chooses two bits $x_0, x_1$ to send to Alice, and Alice (the receiver) chooses the bit $b$ to receive.
At the end of the protocol, Alice is able to recover $x_b$.
Here, we assume the protocol is able to transmit $x_b$ with probability 1.

Here, we say a (quantum) party is semi-honest (or secure against purified adversaries, analogous to the classical honest-but-curious security), if they follow the protocol (without abort) except that they can purify (without loss of generality) all measurements.
We in addition also allow Bob (the sender) to purify his randomness for $x_0, x_1$ if he was to sample them randomly; on the other hand, we require Alice (the receiver) to specify a classical input to make it easier to define security.
(Looking ahead, this is also needed to invoke the semi-honest inefficient attack \cite{CGS16}.)
At the end, they output their residual state as their view for the distinguisher as usual.

\begin{definition}[Statistical security against semi-honest Alice]
  We say an oblivious transfer protocol is $P_A^*$-secure against semi-honest Alice, if for every bits $b, c$, at the end of the protocol with Alice's input being $b$, a semi-honest Alice's view when $x_b = c$ and $x_{1 - b} = 0$ is at most $P_A^*$-close to that when $x_b = c$ and $x_{1 - b} = 1$ in trace distance.
\end{definition}
\begin{definition}[Statistical security against semi-honest Bob]
  We say an oblivious transfer protocol is $P_B^*$-secure against semi-honest Bob, if for every possible (purified) Bob's inputs (meaning an arbitrary bipartite quantum state where the first part is a qubit indicating the input choice bit), at the end of the protocol, a semi-honest Bob's view when $b = 0$ is at most $P_B^*$-close to that when $b = 1$ in trace distance.
\end{definition}

Computational security against semi-honest Alice and Bob can be similarly defined, except considering these two views to be computationally indistinguishable instead of statistically indistinguishable.

We recall the impossibility due to Chailloux, Gutoski, and Sikora showing that oblivious transfer protocols that are statistically secure against both parties do not exist.
While the ``semi-honest'' definition they have is different from here, we could open the proof and check that the cheating strategies constructed there are indeed semi-honest according to our definition.

\begin{theorem}[{\cite[Theorem 1.1]{CGS16}}]
  \label{thm:sot-impossibility}
  For any oblivious transfer protocol, it holds that $2P_B^* + P_A^* \ge 2$, where Alice chooses the choice bit uniformly at random (classically) and Bob chooses the two bits as uniform superposition ($\frac12\paren{\ket{00} + \ket{01} + \ket{10} + \ket{11}}$), and $P_A^*$ is the best probability that a semi-honest Alice is able to predict Bob's choice correctly, and $P_B^*$ is the best probability that a semi-honest Bob is able to predict both bits being sent by Alice correctly.
\end{theorem}

We briefly recall the cheating strategies constructed in their proof.
Alice's strategy is the following \cite[Section 2.1]{CGS16}: she randomly chooses $b$ and then follows the protocol semi-honestly according to our definition; at the end, she performs a gentle measurement to learn $x_b$ (it is gentle since by completeness she is supposed to be able to learn $x_b$ with almost certainty), and then performs the Helstrom measurement (\Cref{thm:holevo-helstrom}) to learn $x_{1 - b}$.
Similarly, Bob's strategy is the following \cite[Section 2.2]{CGS16}: he follows the protocol semi-honestly, purifying all measurements including the uniform sampling of $x_0, x_1$; at the end, he performs a post-processing to try to guess $b$.
Let the success probability of Alice and Bob be $p_a, p_b$ respectively.
They establish that for these two strategies, $2p_b + p_a \ge 2$, and thus the theorem follows.

\subsection{Equivalence theorem}
\label{sec:main-equivalence}

\ifitcsproc
The following theorem is proven in \cite{BCL22}:
\fi
\begin{theorem}
  \label{thm:equivalence}
  The following assumptions are equivalent.
  \begin{enumerate}
    \item Existence of EFI states.
    \item Existence of statistically binding (canonical-form) commitment schemes.
    \item Existence of commitment schemes.
    \item Existence of semi-honest oblivious transfer.
  \end{enumerate}
\end{theorem}
\ifitcsproc\else
\begin{proof}
  $1 \Rightarrow 2$ is shown in \Cref{lem:qg2commitments}.
  $3 \Rightarrow 4$ is shown in \Cref{lem:comm2sot}.
  $4 \Rightarrow 1$ is shown in \Cref{lem:sot2qg}.
  Finally, $2 \Rightarrow 3$ is trivial.
\end{proof}

\begin{lemma}
  \label{lem:qg2commitments}
  Assuming the existence of pairs of EFI states, there exists statistically binding commitments.
\end{lemma}
\begin{proof}
  Let $A_0, A_1$ be the two quantum channels that generate two parts of the EFI state pair respectively.
  Without loss of generality, we assume that their statistical distance is negligibly close to $1$, in particular, at least $1 - e^{-\lambda/2}$ --- by \Cref{cor:td-amplification}, let $\delta = \TraceDist(A_0(1^\lambda), A_1(1^\lambda))$, then taking $\lambda/\delta$ copies of their outputs suffices since $1/\delta$ by assumption is polynomial; on the other hand, computational indistinguishability still holds by a straightforward hybrid argument.
  
  Let the unitary part of $A_i(1^\lambda)$ be $\hat A_i$ for $i = 0, 1$, which acts on registers $\mathsf{CR}$ where the output register is $\mathsf C$ and the auxiliary register is $\mathsf R$.
  The construction of canonical bit commitment scheme is simply running the unitary part as specified above.
  
  It is easy to see that this scheme satisfies computational hiding since by construction $A_b(1^\lambda) = \Tr_{\mathsf{A}}\mparen{A_b\ketbra00A_b^\dagger}$ and thus Bob after the commitment phase sees exactly $A_b(1^\lambda)$.
  On the other hand, since by \Cref{fact:tdnf},
  \[ F\mparen{A_0(1^\lambda), A_1(1^\lambda)} \le \sqrt{1 - \paren{\TraceDist\mparen{A_0(1^\lambda), A_1(1^\lambda)}}^2} \le \sqrt2 e^{-\lambda/4} \]
  is negligible, it also satisfies honest-binding by Uhlmann's theorem.
  Finally, it is known that honest binding for canonical form commitment schemes do imply the more general statistical binding property \cite[Theorems 2 and 3]{Yan22}.
\end{proof}

To construct semi-honest oblivious transfer from commitments, we need the following lemma by Yan, which is originally developed to use computational binding property of a canonical form commitment in order to construct statistically binding quantum commitments from statistically hiding quantum commitments.
On a high level, it reduces the collapsing property of an honest commitment to the computational binding property.
\begin{lemma}[Computational collapse theorem {\cite[Theorem 8]{Yan22}}]
  Let $Q$ a canonical computationally binding quantum bit commitment scheme.
  Then for every polynomial $m$, and every normalized quantum state $\sum_{s \in \{0, 1\}^m}\alpha_s\ket s\ket{\psi_s}$ (of polynomial length), every projector $\Pi$,
  \[
    \abs{\ \norm{\Pi\sum_{s \in \{0, 1\}^m} \alpha_s \ket s \paren{Q_s\ket0}_{\mathsf{CR}^{\otimes m}}\ket{\psi_s}}^2 - \sum_{s \in \{0, 1\}^m} \norm{\alpha_s\Pi\ket s \paren{Q_s\ket0}_{\mathsf{CR}^{\otimes m}}\ket{\psi_s}}^2}
  \]
  is negligible, given that $\Pi$ does not act on the registers $\mathsf{C}^{\otimes m}$ and is efficient.
\end{lemma}

\begin{lemma}
  \label{lem:comm2sot}
  Assuming the existence of commitment schemes, there exists a two-message semi-honest quantum oblivious transfer.
\end{lemma}
\begin{proof}
  The semi-honest oblivious transfer scheme is the following.
  The sender, on input $x_0, x_1$, sends $\ket{x_0} \otimes H\ket{x_1}$; in other words, the sender encodes $x_0$ in the standard basis and $x_1$ in the Hadamard basis and send these two qubits to the receiver.
  The receiver measures both qubits in standard basis if $b = 0$, or in Hadamard basis if $b = 1$.
  Let $x_0, x_1$ be the measurement outcomes.
  The receiver commits to both of them using the commitment scheme in canonical form\footnote{We use canonical form only because we want the commitment to be non-interactive. Even if it is an interactive quantum commitment, everything else could still be extended minus the round complexity.}, and outputs $x_b$.
  
  Semi-honest security against semi-honest sender is easy to see, by a simple hybrid argument invoking the hiding property of the commitment twice, replacing each commitment to committing to 0; furthermore, this reduction even extends if the sender chooses an arbitrary purified input.
  Semi-honest security against receiver follows immediately from collapsing.
  Formally, without loss of generality, we assume the choice bit $b = 0$ and thus the receiver should measure in the standard basis.
  As the goal of the adversary is to extract $x_1$, we can for simplicity remove $x_0$ from the view.
  The second qubit (denoted by register $\mathsf X$) the receiver sends is $\frac1{\sqrt2}\paren{\ket0 + \paren{-1}^{x_1}\ket1}$.
  After the purified protocol concludes, the probability of an efficient distinguisher outputting 1 is
  \[
    \norm{\Pi \frac1{\sqrt 2}\sum_{y \in \{0, 1\}} \paren{-1}^{x_1y} \ket{y} \paren{Q_y\ket0}_{\mathsf{CR}}\ket{\psi_y}}^2
  \]
  for some auxiliary states $\ket{\psi_0}, \ket{\psi_1}$ (possibly containing the distinguisher's advice and auxiliary registers), where $\Pi$ denotes the projector for the distinguisher outputting 1 acting on everything but the $\mathsf C$ register.
  By the computational collapse theorem, this is negligibly close to
  \[
    \frac12\sum_{y \in \{0, 1\}}\norm{\Pi \ket{y} \paren{Q_y\ket0}_{\mathsf{CR}}\ket{\psi_y}}^2,
  \]
  which is independent of $x_1$.
  This concludes the proof.
\end{proof}

We remark that we construct semi-honest OT here instead of considering the CLS scheme directly because (1) semi-honest OT suffices for the following lemma but also (2) as observed by Yan~\cite{Yan22}, it is not clear whether the CLS scheme indeed satisfy malicious (indistinguishability) security when instantiated with a commitment scheme that is only computationally binding but not statistically binding nor extractable, due to the difficulties in using computational binding for a quantum commitment scheme.

Also this semi-honest OT protocol fully intentionally ends with a receiver message to erase the information about the other bit from the distinguisher.
For a classical OT, if the last message comes from a receiver, it can always be removed without impacting the scheme (semi-honest or malicious); this is certainly not the case here.

Finally, this protocol has the minimum round complexity for semi-honest OT.
This can be seen from the following argument.
Even with trusted setup, 1-message protocols are impossible: the message has to be sent by the sender as otherwise the receiver cannot recover the output; however, this means that a semi-honest receiver can always extract both bits efficiently via correctness and gentle measurement.
Without trusted setup, 1-round protocols (two parties exchanging a single message simultaneously) are also impossible: this can be seen as the receiver message is useless and thus reduces to the impossibility above.
Finally, if the sender and the receiver share two EPR pairs, then there is a variant of the protocol above that is 1-round: the receiver is the same but acts on his halves of EPR pairs instead of sender's message, and the sender measures his part of EPR pairs, the first qubit in standard basis and the second qubit in Hadamard basis, obtaining $y_0, y_1$, and send $x_0 \oplus y_0, x_1 \oplus y_1$ to the receiver.
The correctness and security of this protocol can be argued with the same proof techniques as above.

\begin{lemma}
  \label{lem:sot2qg}
  Assuming the existence of semi-honest oblivious transfer, there exists a pair of EFI states.
\end{lemma}
\begin{proof}
  Let us consider both parties to be semi-honest, i.e.\ purifying all the measurements.
  Then their final state would be given by $U_{x_0, x_1, b}\ket0_{\mathsf{AB}}$, for some efficient unitaries $U_{x_0, x_1, b}$ for $x_0, x_1, b = 0, 1$, and at the end Alice holds register $\mathsf A$ and Bob holds register $\mathsf B$.
  By \Cref{thm:sot-impossibility}, we know that either $P_A^* \ge \frac23$ or $P_B^* \ge \frac23$ for every security parameter $\lambda$.
  In particular, either a semi-honest Alice who chooses the choice bit uniformly at random could achieve $P_A^*$, or a semi-honest Bob who chooses the two input bits to be uniform superposition could achieve $P_B^*$.
  
  Fix any security parameter $\lambda$.
  If $P_A^* \ge \frac23$, we construct EFI generators $G_y$ for $y = 0, 1$ that outputs
  \[ \frac12 \cdot \Tr_{\mathsf{B}}\mparen{\E_x\mbracket{\ketbra00 \otimes U_{x, y, 0}\ketbra00U_{x, y, 0}^\dagger + \ketbra11 \otimes U_{y, x, 1}\ketbra00U_{y, x, 1}^\dagger}}. \]
  Note that $G_y$ is exactly Alice's view when she chooses a random choice bit (and remembers the bit in the first register), and Bob sends $y$ in the other slot not chosen by Alice, and thus computational indistinguishability follows by semi-honest security of oblivious transfer.
  On the other hand, they are statistically far since there exists a distinguisher that achieves advantage at least $\frac16$ by our assumption.

  Otherwise if $P_B^* \ge \frac23$, we construct EFI generators $H_b$ for $b = 0, 1$ that outputs $\Tr_{\mathsf{A}}\mparen{\ketbra{\phi_b}{\phi_b}}$, where $\ket{\phi_b}$ is defined to be
  \[ \frac12\sum_{x, y}\paren{U_{x, y, b}\ket0}_{\mathsf{AB}} \otimes \ket x \ket y. \]
  Note that $H_b$ is exactly Bob's view when Alice wants to choose the slot $b$, and thus computational indistinguishability follows by semi-honest security of oblivious transfer.
  On the other hand, they are statistically far since there exists a distinguisher that achieves advantage at least $\frac16$ by our assumption.
  
  Putting everything together, we get that $G_0 \otimes H_0$ and $G_1 \otimes H_1$ generate EFI pairs: statistical farness follows as at least one of $G_b, H_b$ have trace distance $\ge \frac16$ for all security parameters, and computational indistinguishability follows from a direct hybrid argument.
\end{proof}
\fi
\section{Dichotomy for secure two party computations}
\label{sec:2pc}

In this section, we study secure two party computations with quantum parties for classical functionalities.
A secure two-party computation protocol consists of two (interactive uniform quantum) algorithms $A, B$, where they receive (implicitly)
 the security parameter $\lambda$ and their respective inputs $a, b$, take turns to run and exchange a message register back and forth; in the end, we denote their joint state as $\IP{A, B}(a, b)$.
We can also denote Alice's state to be $\IP{A, B}(a, b)_A$ and Bob's to be $\IP{A, B}(a, b)_B$.
Without loss of generality, we consider the protocol so that only Bob gets the output and otherwise they do not learn any other information \cite[Definition 1]{BeimelMM99}.
In this case, $\IP{A, B}(a, b)$ would simply be $(\bot, z)$ as Alice outputs nothing and Bob outputs the evaluated result $z$.
In particular, this evaluated result could be the output of any efficient quantum channel \cite{DupuisNS12}.
We can also define the output state $\IP{A^*, B}(a, b)$ for a malicious Alice (and analogously for Bob), where the malicious Alice can output anything she wants and not necessarily $\bot$.

We now describe the definition for malicious simulation security.
\begin{definition}[Malicious simulation security]
  Let $f = (f_\lambda)_\lambda$ be a quantum channel computable by a polynomial-size quantum circuit.
  A protocol computing $f$ satisfies malicious simulation security for Alice, if the following holds.
  For any (malicious) QPT algorithm $A^*$, there exists a QPT simulator $S$ such that for any QPT distinguisher $D$, there exists a negligible function $\varepsilon$ such that for all security parameter $\lambda$, non-uniform bipartite advice state $\rho_{AD}$, and Bob's input $b$ (permissible by $f_\lambda$),
  \[
    \abs{\Pr\mbracket{D(\IP{A^*, B}(\rho_A, b), \rho_D) = 1} - \Pr\mbracket{D(S_f(\rho_A, b), \rho_D) = 1}} \le \varepsilon(\lambda),
  \]
  where $S_f(\rho_A, b)$ is the following algorithm:
  \begin{itemize}
    \item The two-stage algorithm $S(1^\lambda, \rho_A)$ is run, which outputs some $a^*$.
    \item Compute $(z_a, z_b) \gets f_\lambda(a^*, b)$ to be the output of $f$. (In our setup, $z_a = \bot$ but $z_b$ is the actual output.)
    \item Finish executing $S$ with input $z_a$, which in the end outputs a certain state $\sigma$.
    \item Output $(\sigma, z_b)$.
  \end{itemize}
  
  Malicious simulation security for Bob can be defined in the same way as above, except exchanging the role of Alice and Bob.
\end{definition}

We say the malicious simulation security is statistical if it holds even against any unbounded algorithms $A^*$ and $D$, and in this case there need not be a running time bound on the simulator.

In this work, we focus on secure two-party computations although the consequences also generalize to secure multi-party computations where possibly more than two parties are involved and all of them could receive outputs.
We refer the readers to the prior work \cite{BartusekCKM21a} for related literature.

Combining our equivalence theorem from before and existing work constructing one-sided statistically secure 2PC from statistically binding (quantum) commitments \cite{BartusekCKM21a,AQY21,WolfW06}, we immediately get the following corollary.

\begin{corollary}
  Assuming EFI state pairs exist, then any $\cPpoly$ functionalities can be computed with full malicious security and one-sided statistical security.
\end{corollary}

For the rest of the section, we show EFI states are also implied by non-trivial 2PC protocols.
For that purpose, we focus on 2PC protocols for finite functionalities.
By ``finite'', we mean that the function to be computed is a fixed-size function independent of the security parameter, say Yao's millionaires' problem.

\begin{definition}[Insecure minor]
  Let $S_1, S_2, S_3$ be finite sets and $f: S_1 \times S_2 \to S_3$ be a (finite) function.
  Then we say $f$ contains an insecure minor, if there exists $x_0, x_1 \in S_1$ and $y_0, y_1 \in S_2$ such that $f(x_0, y_0) = f(x_1, y_0)$ and $f(x_0, y_1) \neq f(x_1, y_1)$.
\end{definition}

\begin{lemma}[{\cite[Claim 1]{BeimelMM99}}]
  \label{lem:trivial2pc}
  If a function $f(\cdot, \cdot)$ does not contain an insecure minor, then there is a classical one-message perfectly secure computation protocol for $f$.
\end{lemma}

\begin{lemma}
  \label{lem:nontrivial2pc}
  If a function $f(\cdot, \cdot)$ contains an insecure minor, then we can build an semi-honest OT protocol from an semi-honest secure computation protocol for $f$.
\end{lemma}
\begin{proof}
  This essentially follows from the work of Beimel, Malkin, and Micali~\cite[Claim 3]{BeimelMM99}.
  Since the original proof is black-box, it also immediately generalizes to our setting when the parties are quantum.
  As the precise definition of semi-honest is different in our case, we give the proof for completeness.
  
  Let $x_0, x_1, y_0, y_1$ be the values guaranteed by the insecure minor, and let $\Pi_f$ be the semi-honest secure computation protocol for $f$.
  The (semi-honest) oblivious transfer protocol described in that proof works as follows:
  \begin{itemize}
    \item (Recall) Alice gets as input $a_0, a_1$ and Bob gets as input a choice bit $b$.
    \item Execute $\Pi_f$ on input $x_{a_0}, y_{1 - b}$, and Bob gets output $z_0$.
    \item Execute $\Pi_f$ on input $x_{a_1}, y_b$, and Bob gets output $z_1$.
    \item Bob outputs $0$ if $z_b = f(x_0, y_1)$, otherwise Bob outputs $1$.
  \end{itemize}
  Correctness follows directly since the construction is black-box.
  A semi-honest Alice's view only consists of her semi-honest view from two protocol executions, and thus Alice does not learn anything about $b$.
  Similarly, a semi-honest Bob's view only consists of his semi-honest view from two protocol executions, and thus semi-honest security also follows from the property of insecure minor and Alice's privacy against semi-honest Bob for $\Pi_f$.
\end{proof}

Combining this with \Cref{thm:sot-impossibility}, we immediately get the following:
\begin{corollary}
  If a function $f(\cdot, \cdot)$ contains an insecure minor, then $f$ cannot be computed by statistically-secure semi-honest protocols.
\end{corollary}

Combining our equivalence theorem with \Cref{lem:trivial2pc,lem:nontrivial2pc}, we obtain the following dichotomy theorem.
We shall remark that this theorem, similar to the classical proof \cite{BeimelMM99}, is non-black-box in the use of the functionalities \cite{CrepeauK88,Kil91,KKMO00} due to the use of the equivalence theorem.
\begin{theorem}
  \label{thm:sh2pc2efi}
  If there is a semi-honest two-party secure computation protocol for a classical finite functionality $f(\cdot, \cdot)$, then either $f$ can be computed perfectly securely in a single message or EFI states exist.
\end{theorem}

\section{Quantum computational zero knowledge proofs}

The existence of one-way functions implies that  all of $\cPSPACE$ admit a computational zero knowledge proof \cite{Ben-OrGGHKMR88,Sha92}, and $\cNP$ admits computational zero-knowledge proofs where proofs can be efficiently generated given a witness for membership \cite{GMW86}. Furthermore, the existence of computational zero-knowledge proofs for {\em any} non-trivial (i.e., average-case easy) language implies the existence of (infinitely-often) one-way functions \cite{OW93}.
In this section, we use our equivalence theorem to establish the quantum analogue of these classical results.

\begin{definition}
  A language $L$ is in $\cQCZK$ if there is a (quantum) interactive protocol between an unbounded prover and a QPT verifier (specified by an interactive quantum Turing machine $V$) such that the following holds:
  \begin{itemize}
    \item Completeness: For any $x \in L$, there is an unbounded prover strategy $P$ that would make $V$ accept with probability at least $1 - 2^{|x|}$.
    \item Soundness: For any $x \not\in L$ and any unbounded prover strategy, $V$ accepts with probability at most $2^{|x|}$.
    \item Computational zero knowledge: For any malicious QPT verifier $V^*$, there exists a QPT simulator $S$ such that for any QPT distinguisher $D$ and non-uniform bipartite advice state $\rho_{AD}$, there exists a negligible function $\varepsilon$ such that for any $x \in L$,
      \[ \abs{\Pr\mbracket{D\mparen{\IP{P, V^*(\rho_A)}(x, x)_{V^*}, \rho_D} = 1} - \Pr\mbracket{D\mparen{S(x, \rho_A), \rho_D} = 1}} \le \varepsilon(|x|). \]
  \end{itemize}
\end{definition}
We can also consider a (much) weaker variant of this zero knowledge requirement called computationally zero knowledge against purified verifiers \emph{with abort} (or ``honest verifier''), where we restrict the malicious $V^*$ to only purifying his state and aborting after any fixed number of rounds\footnote{Formally, we require that for a $k$-round protocol, there exists a simulator $S$ whose output is computationally indistinguishable to $\rho_1 \otimes \cdots \otimes \rho_k$, where $\rho_i$ for $i \in [k]$ is the purified verifier's state immediately after receiving $i$-th message from the prover.}.
We call the corresponding class $\cQCZKHV$, similar to $\cQSZKHV$ with respect to $\cQSZK$ for statistical zero knowledge \cite{Wat02}.

Note that here, unlike semi-honest OT where we disallow a semi-honest party to prematurely abort, here we must allow a purified verifier to abort prematurely (this makes the complexity class smaller).
This is because otherwise we can even show the corresponding class (for even quantum \emph{perfect} zero knowledge) is trivially equal to $\cIP$, by simply asking the $\cIP$ verifier at the end to destroy all the other information he has learned by measuring them in Hadamard basis and then returning them to the prover.

\begin{fact}[{\ifitcsproc\cite{Sha92,Wat02,JJUW11,Kobayashi08}\else\cite{Sha92,Wat02,JJUW11} and \Cref{thm:qczkhv-qczk}\fi}]
  $\cBQP \subseteq \cQSZK = \cQSZKHV \subseteq \cQCZK = \cQCZKHV \subseteq \cQIP = \cIP = \cPSPACE$.
\end{fact}

\ifitcsproc\else
\subsection{With EFI, everything provable is provable in QCZK}
\fi
We first consider $\cQCZK$ protocols with efficient provers (in which case the largest complexity class we can consider is $\cQMA$), and then move on to $\cQCZK$ with inefficient provers.
The work by Broadbent and Grilo on $\cQCZK$ \cite{BroadbentG20} show how to build a commit-and-open zero knowledge protocol for $\cQCZK$ using a commitment scheme, and thus combining it with a quantum commitment scheme, we get a $\cQCZK$ proof, but it requires multiple copies of the advice to boost the soundness to negligible.
We  strengthen this to show that we can achieve the same thing from the same assumption, but with a single copy of the witness. 
\ifitcsproc
Overall, the following theorems are proven in  \cite{BCL22}:

\begin{theorem}
  \label{thm:efi2qczk-efficient}
  If EFI states exist, then  $\cQCZK = \cQIP$.  Furthermore, any language in  $\cQMA$  has a quantum  computational zero  knowledge proof with negligible soundness and with an efficient prover that uses only a single copy of the witness.
\end{theorem}
\else
\begin{theorem}
  \label{thm:efi2qczk-efficient}
  If EFI states exist, then  any language in  $\cQMA$  has a quantum  computational zero  knowledge proof with negligible soundness and with an efficient prover that uses only a single copy of the witness.
\end{theorem}
\begin{proof}
  By our equivalence theorem, we get maliciously-secure quantum 2PC for any \emph{quantum} functionality (any quantum channel) with one-sided statistical security via existing works \cite{DupuisNS12}.
  Formally, the protocol is simply the prover and the verifier engaging in the following secure two-party computation for the following quantum functionality that is statistically secure against the prover:
  \begin{itemize}
    \item Prover's input: The witness state.
    \item Verifier's input: None.
    \item The functionality computes whether the amplified verifier (so that the completeness-soundness gap is exponentially close to 1 \cite{marriott2005quantum}) accepts the witness. If so then output $1$ to the verifier, otherwise output $0$. The prover gets no output.
  \end{itemize}
  Completeness follows from the correctness of 2PC.
  Soundness follows from the statistical security of 2PC; in particular, we can extract the witness used by the prover and have the guarantee that this extracted witness passes verification with probability negligibly different from the success probability of the original malicious prover.
  For zero knowledge, we simply tell the 2PC simulator to simulate the malicious verifier's view given the ideal functionality outputting 1.
  Since the amplified verifier has completeness exponentially close to 1, the output of 2PC is unchanged with overwhelming probability, and thus invoking the security of 2PC, this change is computationally indistinguishable.
\end{proof}

\begin{theorem}
  \label{thm:qczk-qip}
  If EFI states exist, then $\cQCZK = \cQIP$.
\end{theorem}
\begin{proof}
  By our equivalence theorem, we get a maliciously-secure quantum 2PC for any (classical) \emph{reactive} functionality with one-sided statistical security via existing works \cite{CrepeauGT95,IshaiPS08}.
  By reactive, we mean that the functionality can interact with the two parties and keep private states.
  Given any language $L \in \cQIP$, let $\Pi$ be a (classical many-round) Merlin--Arthur interactive proof protocol for $L$ with completeness 1.
  The zero knowledge protocol is simply the prover and the verifier engaging in the following secure two-party computation for the following reactive functionality that is statistically secure against the prover:
  \begin{itemize}
    \item $\Pi$ is executed in the following way: for every prover message, the functionality remembers the message and only sends $\top$ to the verifier; and for every verifier message, the functionality forwards it to the prover.
      (Both parties are able to participate in the protocol as usual since the prover sees all the messages and the verifier only needs to flip random coins.)
    \item At the end, if the verifier for $\Pi$ accepts the transcript, we output $1$ to the verifier; otherwise we output $0$. The prover gets no output.
  \end{itemize}
  Similarly as before, completeness follows from completeness of $\Pi$ and the correctness of 2PC.
  Soundness follows from the soundness of the original $\cIP$ protocol and the statistical security of 2PC; in particular, we can use the 2PC simulator to come up with a malicious prover for $\Pi$ and have the guarantee that this malicious prover has success probability negligibly different from the success probability of the original malicious prover.
  For zero knowledge, we simply invoke the 2PC simulator for the verifier, and tell the simulator that the ideal functionality outputs 1.
  Since the protocol has completeness 1, the output of 2PC is unchanged no matter what the verifier's input is, and thus invoking the security of 2PC, this change is computationally indistinguishable.
\end{proof}

\subsection{EFI pairs are essential for non-trivial QCZK}
\fi
\ifitcsproc
\begin{theorem}
  \label{thm:hardness2qci}
  If there is a language in $\cQCZKHV$ that is hard on average for $\cBQP$ for some $\cBQP$-samplable distribution $D$, i.e.\ any $\cBQP$ algorithm has negligible success probability in deciding $L$ on average over $D$ and the algorithm's randomness,  then EFI state pairs (secure against uniform $\cBQP$ adversaries) exist.
\end{theorem}
\else
\begin{theorem}
  \label{thm:hardness2qci}
  If there is a language in $\cQCZKHV$ that is hard on average for $\cBQP$ for some $\cBQP$-samplable distribution $D$, i.e.\ any $\cBQP$ algorithm has negligible success probability in deciding $L$ on average over $D$ and the algorithm's randomness,  then EFI state pairs (secure against uniform $\cBQP$ adversaries) exist.
  \end{theorem}

As is the classical case \cite{OW93}, the other direction is more involved.
Since we could not use one-way functions like the classical proof, we instead use ideas from the works of Ong and Vadhan~\cite{Vadhan06,OngV08}.
On a high level, they also study the equivalence between zero-knowledge and commitments in the classical setting, but their proof goes through universal one-way hash functions.
We adapt their high-level proof ideas but replace certain ingredients with quantum techniques originally developed by Watrous when studying quantum statistical zero knowledge proofs \cite{Wat02}.
In particular, we are going to take their idea of considering instance-dependent probability ensembles/commitments, and propose the quantum analogue of instance-dependent probability ensembles in the following lemma, showing the equivalence between a language being inside $\cQCZK$ and the existence of an instance-dependent computational indistinguishability for that language\footnote{We only formally show one direction, but the other direction follows the same proof as the classical case via constructing zero-knowledge proofs from instance-dependent commitments.}.

\begin{lemma}
  \label{lem:qczk2idstates}
  If a language $L$ admits a $\cQCZKHV$ proof, then there are instance-dependent mixed states $\{\xi_{b, x}\}_{b, x}$ for $L$ where $b = 0, 1$ and $x \in \{0, 1\}^*$ such that:
  \begin{itemize}
    \item There is a uniform QPT procedure that on input $b, x$ generates $\xi_{b, x}$.
    \item For every nonuniform QPT distinguisher $D$, there is a negligible function $\varepsilon$ such that for all $x \in L$,
    \[\abs{\Pr\mbracket{D(x, \xi_{0, x}) = 1} - \Pr\mbracket{D(x, \xi_{1, x}) = 1}} \le \varepsilon(|x|).\]
    \item There is some constant $c$ such that for every $x \not\in L$, $\TraceDist(\xi_{0, x}, \xi_{1, x}) \ge c$.
  \end{itemize}
\end{lemma}
\begin{proof}
  The mixed states that we consider is due to Watrous~\cite[Theorem 7]{Wat02}, who introduced these states while finding a complete problem for $\cQSZK$.
  We are given a $k$-round computational zero knowledge protocol, where each round consists of the verifier sending a message followed by the prover sending a response.
  In Watrous's work, $\gamma_0$ and $\gamma_1$ are defined to be $\gamma_0 = \rho_1 \otimes \cdots \otimes \rho_k$ and $\gamma_1 = \xi_1 \otimes \cdots \otimes \xi_k$, where $\rho_i$ corresponds to the simulated purified verifier's state if he aborts at the end of round $i$ and then trace out the prover's message, and $\xi_i$ corresponds to the simulated purified verifier's state if he aborts immediately before sending the $i$-th round's message and then trace out the verifier's message (more formally, this is taken to be the simulated output after $(i - 1)$-th message from the prover, apply the verifier's action in $i$-th round and then trace out the message register).
  Efficient generation is immediate and trace distance for NO instances follows from the soundness of the protocol as in the case for $\cQSZK$ \cite[Lemma 8]{Wat02}.
  
  Finally, computational indistinguishability for YES instances follows from a straightforward hybrid argument, invoking the computational zero knowledge property.
  This is because in the honest execution (for a YES instance), the verifier's (unsimulated) states corresponding to $\rho_i$ and $\gamma_i$ are identical since they only differ by an action from the prover, which does not act on the verifier's private register that we are not outputting.
  More formally, let $\gamma_2 = \psi_1 \otimes \cdots \otimes \psi_k$, where $\psi_i$ is the purified verifier's state after $i$ rounds of interaction with the real prover.
  By semi-honest zero knowledge, we immediately have that $\gamma_0$ is computationally indistinguishable to $\gamma_2$.
  Invoking semi-honest zero knowledge again, we have that $\gamma_1$ is computationally indistinguishable to a state that is identical to $\gamma_2$, by our discussion above.
  This completes the proof via hybrid argument.
\end{proof}

We first sketch how this object is powerful enough for us to give the following alternative proof for an unconditional result about $\cQCZKHV$ originally established by Kobayashi~\cite{Kobayashi08}.

\begin{theorem}[{\cite{Kobayashi08}}]
  \label{thm:qczkhv-qczk}
  $\cQCZKHV = \cQCZK$ holds unconditionally.
\end{theorem}
\begin{proof}[Proof sketch]
  By definition, $\cQCZK \subseteq \cQCZKHV$.
  Using the lemma above with \Cref{lem:qg2commitments}, for any language $L \in \cQCZKHV$, we obtain instance-dependent quantum commitments, which is statistically binding when $x$ is not in the language and computationally hiding when $x$ is in the language.
  This is sufficient to obtain an instance-dependent 2PC for reactive functionalities via known compilers \cite{BartusekCKM21a,AQY21,WolfW06,CrepeauGT95,IshaiPS08}, where it is statistically secure against the prover when $x$ is not in the language and computationally secure against the verifier when $x$ is in the language.
  This can be then plugged into \Cref{thm:qczk-qip} to give a $\cQCZK$ proof for $L$.
\end{proof}

\begin{proof}[Proof of \Cref{thm:hardness2qci}]
  Assume that there exists $L$ that admits a $\cQCZK$ proof but $(L, D)$ is hard on average against $\cBQP$ for some efficiently samplable distribution $D$.
  By \Cref{lem:qczk2idstates}, we get instance-dependent mixed states.
  
  Since $D$ is hard-on-average, it is easy to see that YES and NO instances are of at least 1/3 fraction for all sufficiently large instance lengths, as otherwise a trivial machine that outputs a constant decides this language with a noticeable advantage (infinitely often).
  Given that this is the case, it is easy to see that the mixed states $(x, \xi_{b, x})$ taken average over this distribution for $x$ is still statistically far.
  
  We now assume for contradiction that this EFI state pair does not satisfy computational indistinguishability, i.e.\ there is a $\cBQP$ algorithm $M$ that distinguishes $(x, \xi_{0, x})$ from $(x, \xi_{1, x})$.
  In particular, without loss of generality that it outputs $1$ with probability $c(|x|)$ when $b = 1$ but $s(|x|)$ when $b = 0$ over $|x|$, where $c - s$ is noticeable.
  
  We now give a QPT algorithm for $L$: on input $x$, generate the instance-dependent mixed states with $b$ uniformly chosen from random; if $M$ correctly predicts $b$ then reject, otherwise accept.
  When the input is a YES instance, we accept correctly with probability within $\bracket{\frac12 - \varepsilon(|x|), \frac12 + \varepsilon(|x|)}$ for some universal negligible function $\varepsilon$ as otherwise we can use this algorithm to break computational indistinguishability of the instance-dependent mixed states.
  Since the algorithm predicts correctly with non-negligible probability over the entire domain of $D$, it must be the case that almost all the prediction advantages are from NO instances\footnote{We cannot argue this fact directly from the non-existence of EFI states since the distribution conditioning on the input being a NO instance might not be efficiently samplable.}.
  Therefore, for NO instances, this algorithm correctly rejects with probability noticeably more than $\frac12$.
  This breaks the average-case hardness of $(L, D)$, a contradiction.
\end{proof}

\begin{theorem}
  \label{thm:aiefi}
  If $\cBQP \neq \cQCZK$ then auxiliary-input EFI state pairs exist, where the definition of auxiliary-input EFI state pairs is the following:
  \begin{enumerate}
    \item \textbf{Efficient generation}: There exists a uniform QPT quantum algorithm $A$ that on input $(x, b)$ outputs the mixed state $\xi_{b, x}$.
    \item \textbf{Statistically Far}: There exists an infinite set $S \subseteq \{0, 1\}^*$ and some polynomial $p$, such that for any $x \in S$, we have that $\TraceDist\mparen{\xi_{0, x}, \xi_{1, x}} \ge 1/p(|x|)$.
    \item \textbf{Computational Indistinguishability}: For any uniform QPT distinguisher $D$ and every polynomial $q$, there exists $x \in S$ such that $\abs{\Pr[D(x, \xi_{0, x}) = 1] - \Pr[D(x, \xi_{1. x}) = 1]} \le 1/q(|x|)$.
  \end{enumerate}
\end{theorem}
\begin{proof}
  The proof works similarly as before.
  Given a $\cQCZK$ language $L$, we start by constructing instance-dependent mixed states $\xi_{b, x}$.
  Consider $S = \bar L$.
  If $S$ is finite then $L$ is trivial then we are done.
  Otherwise, by the non-existence of auxiliary-input EFI states, we get that there exists a uniform $\cBQP$ distinguisher $D$ such that for every $x \in S$, the $\cBQP$ distinguisher distinguishes these two distributions with inverse polynomial probability.
  Using the same $\cBQP$ decider as the proof above, we get that this decider (after parallel repetition and majority vote) is correct with probability at least $2/3$ except on finitely many instances.
  Therefore, $L \in \cBQP$.
\end{proof}
\fi
\ifitcsproc\else
\section*{Acknowledgements}
\acks
\fi
\fi
\fi

\iflipics
  \bibliography{bib}
\else
  \printbibliography
\fi

\end{document}